\documentclass[prodmode,acmtosn]{acmsmall}

\usepackage{setspace}
\usepackage{srcltx}
\usepackage{amsmath}
\usepackage{amssymb}
\usepackage{graphics}
\usepackage{cite}
\usepackage{psfrag}
\usepackage{epsfig}

\newtheorem{thm}{Theorem}

\newtheorem{lem}{Lemma}

\acmVolume{10}
\acmNumber{1}
\acmArticle{0}
\acmYear{2014}
\acmMonth{0}

\begin{document}

\markboth{A. Eslami et al.}{Results on Finite Wireless Sensor Networks: Connectivity and Coverage}

\title{Results on Finite Wireless Sensor Networks: Connectivity and Coverage}
\author{Ali Eslami, Mohammad Nekoui, and Hossein Pishro-Nik
\affil{University of Massachusetts, Amherst}
Faramarz Fekri
\affil{Georgia Institute of Technology}
}

\begin{abstract}
Many analytic results for the connectivity, coverage, and capacity of wireless networks have been
reported for the case where the number of nodes, $n$, tends to infinity (large-scale networks). The majority of
these results have not been extended for small or moderate values of $n$; whereas in many practical networks,
$n$ is not very large. In this paper, we consider finite (small-scale) wireless sensor networks.
We first show that previous asymptotic results provide poor
approximations for such networks. We provide a set of differences between small-scale and large-scale analysis and propose a methodology for analysis of finite sensor networks. Furthermore, we consider two models for such networks: unreliable sensor
grids, and sensor networks with random node deployment. We provide easily computable expressions for bounds on the coverage and
connectivity of these networks. With validation from simulations, we show that the derived analytic expressions give very
good estimates of such quantities for finite sensor networks. Our investigation confirms the fact that
small-scale networks possesses unique characteristics different from the large-scale counterparts, necessitating the development of a new framework for their analysis and design.
\end{abstract}

\category{C.2.0}{Computer-Communication Networks}{General}

\terms{Data Communications}

\keywords{Wireless sensor networks, finite networks, unreliable grids, random geometric graphs,
 connectivity, coverage}

\acmformat{Eslami, A., Nekoui, M., Pishro-Nik, H., Fekri, F.  2011. Analysis of connectivity and coverage for finite wireless sensor networks.}

\begin{bottomstuff}
The material in this paper was presented in part at the 5th Annual IEEE Conference on Sensor, Mesh and Ad Hoc Communications and Networks, June 2008, and the 4th IEEE International Symposium on Modeling and Optimization in Mobile, Ad Hoc and Wireless Networks, April 2006.

This work is supported by the National Science Foundation, under
grants CCF-0728970 and ECCS-0636569.

Author's addresses: A. Eslami, M. Nekoui, and H. Pishro-Nik, Department of Electrical and Computer Engineering,
University of Massachusetts, Amherst, email:\{eslami, nekoui, pishro@ecs.umass.edu\}; F. Fekri, School of Electrical and Computer Engineering, Georgia Institute of Technology, email:fekri@ece.gatech.edu.
\end{bottomstuff}

\maketitle

\section{Introduction}

In the past, many analytic results on the connectivity, coverage,
and capacity of wireless ad-hoc and sensor networks have been
obtained. In almost all of the results, it is assumed that the
number of nodes, $n$, in the network tends to infinity (large-scale
networks). In other words, these results are asymptotic.
Asymptotic results are very important for two reasons. First, they give us good
estimates for large-scale networks. Second, they show some
fundamental trade-offs in the network. However, in many practical
wireless networks the number of nodes may be limited to a few
hundreds (small-scale/finite networks). As it is shown in this paper,
the asymptotic results cease to be valid for these networks. Thus,
it is very crucial from the practical point of view to analyze finite
networks. These analytic results will essentially help us to
understand, design, and analyze practical wireless networks, and
also to design more suitable communication protocols.

For example, consider the capacity analysis of
wireless networks which has been studied extensively (e.g., in
\cite{maccapacity,guptacapacity2000,mobility01,gupta03,grid:mobicom01,Perevalov03,hybrid03}). Today we have a good
understanding of scaling laws for the capacity of wireless networks.
However, suppose we need to design a wireless sensor network
consisting of an arbitrary deployment of a hundred sensor nodes. Some fundamental questions are
as follows. What is the transport capacity? What are the connectivity and coverage probabilities of such networks? How do network parameters
such as the communication radius of nodes, number of nodes, and so on,
affect these properties? Unfortunately, the available asymptotic
results fail to give answers to these questions.

To address the aforementioned issues in small-scale networks we need to address some inherent problems. First,
in large-scale networks we use asymptotic estimates that make the analysis much simpler. These estimates are
not available in small-scale analysis. Thus, small-scale analysis is usually more difficult. Second, even if we
can perform the small-scale analysis, we usually obtain very complicated formulas that are not very useful
practically. In this paper, we want to circumvent these problems and provide bounds for small
scale-analysis. In particular, we are looking for easily computable but acceptable estimates for fundamental
network quantities.  The main goal of this paper is to initiate the small-scale analysis of wireless sensor networks. To the best of our knowledge, this is the first work to analytically and systematically study
this issue.

The main idea is the following. The first key point is to aim at simple and very good approximations instead of
trying to find complicated exact formulas. To do so, we first consider the asymptotic analysis. Some of the estimates in the asymptotic analysis are still good for
small-scale networks, while others are not. We identify those which are not valid and replace them with better
estimates.  However, this must be done carefully, in order to obtain simple and easily computable formulas at
the end. Specifically, in this paper we list a few important differences between small-scale and large-scale
analysis.

As a special case of finite sensor networks, we first study unreliable sensor grids in which the sensors are
deployed in a grid and each sensor is active with probability $p$. This probability is used to account for both
sensor failures and sleeping sensors. A fundamental question is that given an area to be protected, how many
sensors should be deployed so that every point in the region is covered by at least one sensor (more generally,
we may require that every point in the region is covered by at least $k$ sensors). Equivalently, one can ask if
$n$ sensor nodes are deployed in an area, what should be the sensing radius of nodes to ensure coverage (or
$k$-coverage)?  The same question can be repeated for other network properties such as connectivity and
diameter. In this paper, we study the behavior of the different parameters in finite sensor grids. We prove that
all graph theoretic properties of these networks such as connectivity, network diameter and capacity, follow a
piecewise constant behavior and this is even true for the coverage which is not a graph-theoretic property. This
result shows a key difference between the behavior of sensor grids and randomly deployed sensor networks, and
has some important implications from the practical point of view: 1. It shows that increasing the communication
and sensing radii does not necessarily improve coverage, connectivity or any other graph-theoretic property,
2. It suggests that we can completely determine the behavior of a vast class of network properties by knowing
their values for only a finite number of points. We then find simple lower and upper bounds for the k-coverage
probability of sensor grids and show that these bounds are adequately close to the real value, as an estimate of
the coverage probability.

Next, we consider finite sensor networks in which nodes are randomly distributed in the unit square. We study
$k$-connectivity and coverage of these networks. We give several results pertaining to these properties.
We first show that the previous asymptotic results on coverage and $k$-connectivity are not accurate for the
finite case. We then provide a very simple formula for the $k$-connectivity probability of finite sensor networks and show that the formula is very precise. We also study the coverage probability of random networks where we prove simple lower and upper bounds for the coverage probability.

The remainder of the paper is organized as follows: Next section provides an overview of the related work.
In Section \ref{sec:fund}, we study connectivity and coverage of finite sensor grids.
Section \ref{sec:general} investigates the
fundamental properties of random sensor networks such as connectivity and coverage. Finally, Section
\ref{sec:conclusion} concludes the paper.

\subsection{Related Work} \label{sec:related}
Related problems have been studied in the context of random graph theory \cite{RandomGraphs}, continuum
percolation and geometric probability \cite{continuumpercolation,Randomgeographspenrose}, and the study of
wireless network graphs \cite{gupta98critical,guptacapacity2000,kumar2004neighbors,Franceschetti2003covering,Franceschettiunreliable,unreliable2003sensor,Dubhashisubmitted,Fault03,asymptotic04,kcoverage}. In random
graph theory, the model $G(n,p)$ is extensively studied, in which edges appear in a graph of $n$ vertices with
probability $p$ independent of each other. In continuum percolation theory, usually infinite graphs on
$\mathbb{R}^d$ are studied. Finally, in geometric probability and the study of graphs of wireless networks,
large-scale graphs over the plane are usually studied.

In \cite{frances,francescritical}, the authors studied connectivity and critical node life-time for a model of random networks in which the density of nodes is kept constant while the area of interest tends to infinity. Furthermore, the throughput scaling of wireless relay networks is studied in \cite{throuscale} for this model. However, the results in these papers are all based on asymptotic analyses and their method can not be applied to the case of finite networks, i.e. networks with finite number of nodes (e.g. less than 1000) on a finite plane. In the analysis of these networks, boundary effects and constant factors (see section \ref{sec:need}) cannot be neglected as can be for the case of asymptotic analysis.

The connectivity and $k$-connectivity of large-scale wireless networks have been investigated in \cite{gupta98critical}, \cite{Fault03}, \cite{asymptotic04}, \cite{mysecon}, and \cite{Thiran02}.
In
\cite{convscap04}, the trade-off between connectivity and capacity of
dense networks was examined. The transport, information theoretic, and MAC layer capacities have extensively been investigated (see for example
\cite{maccapacity,guptacapacity2000,mobility01,gupta03,grid:mobicom01,Perevalov03,hybrid03}.
The grid model for sensor networks has also been investigated. In particular, connectivity, coverage, and
diameter of sensor grids were studied in \cite{unreliable2003sensor}. In \cite{kcoverage2} and \cite{covdim}, the $k$-coverage problem for sensor grids and other deployment methods was considered.
The authors in \cite{balister09trap,balister09random} also studied coverage for sensor networks in presence of failures and placement errors.
However, almost all previous analytical results are asymptotic since they consider large-scale networks.

Analysis of wireless networks with modest number of nodes has generated a lot of interest
in the recent past \cite{bai06,connectivity_commletters_02,connectivity_commletters_06,rand_mob_wireess_wiopt06,link_prob_mass_04,finiteonedim06}. In \cite{connectivity_commletters_02}, the authors investigated the
problem of connectivity for one-dimensional
networks (line networks). Using probabilistic methods, they obtained the exact
formulation for the probability of connectivity. The author of \cite{connectivity_commletters_06}
presented corrections and extensions to
\cite{connectivity_commletters_02}. It is noted that both of the
above cases considered a line network, and the extension to
two-dimensional networks was achieved by obtaining a loose bound
using the results from the former case. In \cite{finiteonedim06},
the authors also consider the line network and obtain connectivity
results for one-dimensional networks.
The threshold phenomena for finite wireless networks on a line is studied in \cite{eslamitcomfinite10}. The authors also find lower and upper bounds on the MAC-layer capacity for such networks.
It should be noted that the main challenges
in finite analysis arise in the two dimensional case. In
\cite{rand_mob_wireess_wiopt06}, mobility and more realistic models
were examined. The authors obtained results on the connectivity for
both finite and asymptotic cases in one-dimensional networks. In
\cite{link_prob_mass_04}, some simple local network characteristics
such as the link probability (occurrence of a link) and average node degrees are studied. The
paper also obtains formulas for the average covered area.
In \cite{balister07}, connectivity and coverage are studied for networks on a thin strip of finite length. The authors provide reliable density estimates for achieving coverage and connectivity, assuming a Poisson distribution for the nodes.

\section{Fundamental Properties of Finite unreliable Sensor Grids} \label{sec:fund}

In this section we present properties of finite unreliable sensor grids. In particular, we prove that a large
class of network properties such as connectivity, coverage, and capacity can be represented as a piecewise
constant function of the communication and sensing radii, $r_t$ and $r_s$, respectively. We also discuss the
implications of this result and show the importance of boundary effects in finite networks. We then find an
upper bound for coverage which can be used to approximate the exact value of the coverage.


Here, we consider the sensor network model introduced in \cite{unreliable2003sensor}. In particular, it is
assumed that $n$ sensor nodes are arranged in a grid over a square region of unit area. This region is called
the deployment region and it is assumed to be the unit square centered at the origin.
Such a grid is depicted in Fig. \ref{fig:SenGrid}.
We show the deployment region by $S_0$. The separation between adjacent nodes is assumed to be $\frac{1}{\sqrt{n}}$ units. Each sensor
node can detect events within some distance from it, called the sensing radius $r_s$. Each sensor is active with
probability $p$ independently from other nodes. The transmission radius of each node is assumed to be $r_t$. In
other words, if the distance between two sensor nodes $u$ and $v$ is less than $r_t$, then they can communicate
with each other, thus the edge $\{u,v\}$ belongs to edges of the graph. It is worth noting that our results
apply to any deterministic placement of finite sensor networks and also any finite deployment region with smooth
boundaries. However, for simplicity, we consider the above grid model in this paper. We are interested in
connectivity and coverage. In particular, we assume $p_{disc}(n,p,r_t)$ is the probability that the sensor grid
with parameters $n$, $p$, and $r_t$ constructs a disconnected graph. We also assume that $p_{cov}(n,p,r_s,k)$ is
the probability that each point of the unit square (the deployment region) is covered by at least $k$ sensors in
the sensor grid with parameters $n$, $p$, and $r_s$. Thus $p_{cov}(n,p,r_s,1)$ is the probability that the whole
unit quare is covered by the sensor nodes.

\begin{figure}[t]
\centering
\includegraphics[width=2.7 in, height=2.5 in]{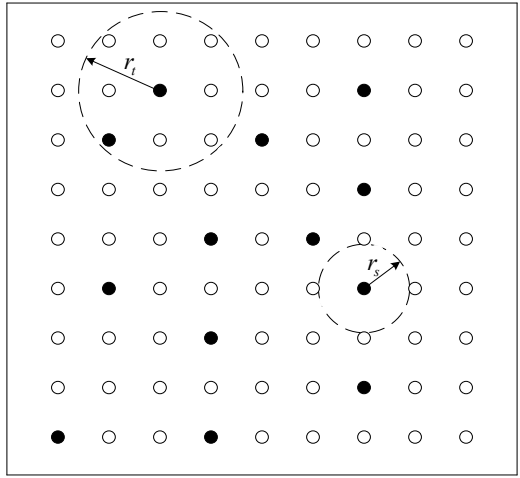}
\caption{A sensor grid is shown where the active sensors are shown by black nodes. communication and sensing radii, $r_t$ and $r_s$, are also shown.}
\label{fig:SenGrid}
\end{figure}

\subsection{Sensor Grids: Asymptotic versus Finite Analysis} \label{sec:needgrid}
We now present some evidence to show that previous asymptotic
results diverge significantly from actual values for finite grids.
To show this, we consider connectivity and coverage. Let us first
consider coverage. The asymptotic coverage probability
$p_{cov}(n,p,r_s,k)$ has been found in \cite{kcoverage2}. In particular
the following fundamental result has been obtained in
\cite{kcoverage2}.
\begin{thm} \label{KLB}(Kumar, Lai, and Balogh 2008)
Let $\epsilon$ be an arbitrary constant positive real number and $k$
be a constant positive integer. Then for $n$ chosen large enough we have the following two cases.
\begin{itemize}
   \item If $r_s(n) \geq \sqrt{\frac{(1+\epsilon) \log (np)}{\pi np}}$,
   then the unit square is almost always $k$-covered completely, i.e.,
    $p_{cov}(n,p,r_s,k)=1-o(1)$.
    \item If $r_s(n) \leq \sqrt{\frac{(1-\epsilon) \log (np)}{\pi np}}$, then
    $p_{cov}(n,p,r_s,k)=o(1)$.
\end{itemize}
\end{thm}
Using simulations, authors of \cite{kcoverage2} have shown that this
theorem results in accurate estimation of $p_{cov}(n,p,r_s,k)$, when
$n$ is large (say $n>10000$). Thus the theorem is very useful in the
design of large-scale sensor networks. Let us now consider a sensor
grid consisting of $100$ unreliable sensor nodes with $p=0.2$. If we
want to use the asymptotic result for this network, choosing
$\epsilon=0.1$, we conclude that if $r\geq 0.229$ then
$p_{cov}(n,p,r_s,k)\approx 1$ and if $r\leq 0.207$ then
$p_{cov}(n,p,r_s,k)\approx 0$. We have used exhaustive simulations
to obtain an accurate estimate of $p_{cov}(n,p,r_s,k)$. In
Figure \ref{fig:coverageasymp}, we compare the results obtained by
exhaustive simulations and Theorem \ref{KLB}. It is observed that
the two results differ considerably. For example, at $r_s=0.25$, the
asymptotic result predicts that the unit square is covered with
probability close to one. However, simulations show that this
probability is only $p_{cov}(n=100,p=0.2,r=0.25,k=1)=0.018$. It is
clear that for this network the asymptotic analysis cannot provide
results that are sufficiently accurate.
Figure \ref{fig:coverageasympn100p2k2} shows that the same situation exists when we consider $k$-coverage for $k>1$. Thus, it is very important to provide finite-size analysis. We also observe that the coverage
probability obtained by simulations shows several discontinuities.
We prove this phenomenon in the section \ref{sec:piecewise}. We performed many simulations for different values of $n$, $p$, and $k$ to further validate
the insufficiency of asymptotic results. However, we omit them for brevity.

\begin{figure}[t]
\centering
\includegraphics[width=4 in, height=3 in]{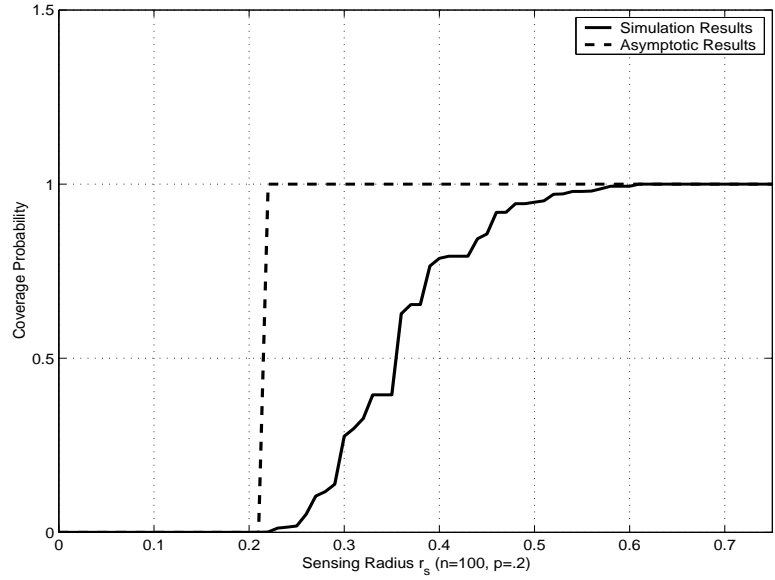}
\caption{Comparison of asymptotic results and the simulation results for the coverage probability of a sensor grid with $p=0.20$. For simulations we picked $n=100$ while the asymptotic result hold when $n$ tends to infinity.}
\label{fig:coverageasymp}
\end{figure}

\begin{figure}[t]
\centering
\includegraphics[width=4 in, height=3 in]{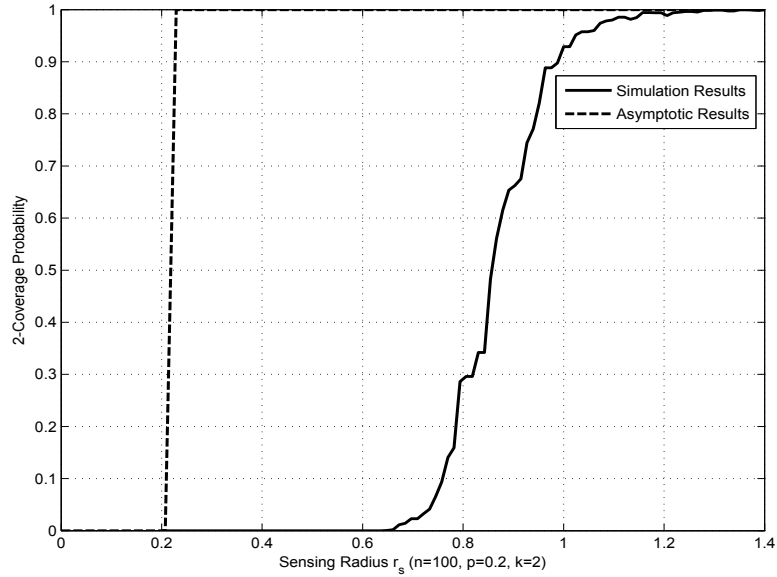}
\caption{Comparison of asymptotic results and the simulation results for 2-coverage probability ($k$=2) of a sensor grid with $p=0.20$. For simulations we picked $n=100$ while the asymptotic result hold when $n$ tends to infinity.}
\label{fig:coverageasympn100p2k2}
\end{figure}

\subsection{Discontinuity in Properties of Sensor Grids}\label{sec:piecewise}
Here we prove that a
vast class of network properties can be represented by piecewise constant
functions of $r_t$ and $r_s$. We stress that the piecewise property is one of the key differences between sensor grids
(deterministic deployment) and randomly deployed sensor networks.  Consider a right-continuous function
$f(x):[0,\infty]\rightarrow \mathbb{R}$. The function $f$ is said to be
piecewise constant if there exists a set of real numbers $0=x_1
<x_2<x_3<...$, and $c_1$, $c_2$, ..., such that $f(x)=c_i$ for all
$x \in [x_i,x_{i+1})$. In this paper we only deal with functions for
which the number of $x_i$'s is finite.

Let $Q$ be a property for sensor grids such as coverage, i.e., we say
that a grid has the property $Q$ if it covers the deployment region.
Coverage is an example of \textit{geometric properties}. Another
category of properties are \emph{graph theoretic properties} such as
connectivity. In particular, any sensor grid with parameters $n$,
$p$, and $r_t$ corresponds to a graph that can be shown by $g(V,E)$ where $V$ and $E$ are the set of vertices and edges in the graph, respectively.
The sensor nodes construct $V$, the set of vertices of the graph.
There exists an edge between two vertices if their corresponding
sensors are within the communication range of each other. Any
property of $g(V,E)$ is a graph theoretic property for the sensor
grid. Thus, two different sensor grids will have the same graph
theoretic properties if they have isomorphic (identical) graphs.
We note that coverage is not a graph theoretic property.

Let $\mathcal{X}=\{X_1,X_2,...,X_k\}$ be a set of points on the plane.
Define $g(\mathcal{X},r_t)$ as the graph obtained by the following method.
The vertices of $g$ are the points in $\mathcal{X}$ and there is an
edge between two vertices $X_i$ and $X_j$, if their distance is less
than or equal to $r_t$. We prove the following theorem.

\begin{thm}\label{thm:piecegraph}
Let $Q$ be a graph theoretic property of sensor grids with
parameters $n$, $p$, $r_t$. Let $n$ and $p$ be fixed numbers and
$p_{Q}(r_t)$ be the probability that the sensor grid with
communication radius $r_t$ has the property $Q$. Then $p_Q$ is a
piecewise constant function. In particular, there exist
$0=r_0<r_1<r_2<...<r_m \leq \frac{1}{\sqrt{2}}$, and $c_1, c_2,...
,c_{m}$ such that $p_{Q}(r_t)=c_i$ if $r_t \in [r_{i-1},r_i)$.
\end{thm}

\begin{proof}
Let $\mathcal{X}=\{X_1,X_2,...,X_k\}$ be the set of points in the sensor grid.
Let also $\mathcal{X}_a \subset \mathcal{X}$ be the set of active
sensors. Assume that $g(\mathcal{X},r_t)$ is the corresponding graph.
Let $p(\mathcal{X}_a)$ be the probability that $\mathcal{X}_a$ is the set of active sensors, then we have
\begin{equation}\label{}
p(\mathcal{X}_a)=p^{\mid p(\mathcal{X}_a)\mid} (1-p)^{n-\mid
p(\mathcal{X}_a)\mid},
\end{equation}
where $\mid p(\mathcal{X}_a)\mid$ is the number of active sensor
nodes. Then
\begin{equation}\label{equ:pq}
p_{Q}(r_t)= \sum_{g(\mathcal{X}_a,r_t)  \textrm{ has } Q}
p(\mathcal{X}_a).
\end{equation}
It suffices to find $0=r_0<r_1<r_2<...<r_m< \frac{1}{\sqrt{2}}$
 such that the network graphs $g(\mathcal{X}_a,r_t)$ remain constant for $r_t \in [r_{i-1}, r_i)$ for any choice of
 $\mathcal{X}_a$ and any $i \in \{1,2,...,m\}$. Let $D=\{d_1,d_2,...
 ,d_l \}$ be the set of distances between the points in
 $\mathcal{X}$, and assume that
 $0<d_1<d_2<...<d_l \leq \frac{1}{\sqrt{2}}$. In our grid model we have $d_i=\sqrt{\frac{i}{n}}$, $i=1,...,l$. Then, the network graph remains the same when $r_t \in [d_i,
 d_{i+1})$ for any $i \in \{1,2,...,m\}$. This is because changing
 $r_t$ within $[d_i,d_{i+1})$ will not add or remove any edges. This
 means that we can choose $r_i=d_i$. Thus $p_{Q}$ in (\ref{equ:pq}) remains constant for $r_t \in [r_{i-1}, r_i)$.
 It is also easy to see that $p_{Q}$ is right-continuous because the edges in the graphs are formed when
 the distance between two nodes is less than or equal to $r_t$. This
 completes the proof.
\end{proof}

Note that the above discussion shows that any graph theoretic
quantity is a piecewise constant function of $r_t$. This includes
diameter of the network, MAC layer capacity \cite{maccapacity}, $k$-connectivity,
etc. We now prove that coverage probabilities are piecewise constant
functions of sensing radius. Note that this cannot be concluded from Theorem \ref{thm:piecegraph},
since coverage is not a graph theoretic property.

\begin{thm}\label{thm:piececoverage}
Consider a sensor grid with parameters $n$, $p$, $r_s$. Let $n$ and $p$ be fixed
numbers. Then $p_{cov}(n,p,r_s,k)$ is a piecewise constant function of $r_s$. In particular, there exist
$0=r_0<r_1<r_2<...<r_m \leq \frac{1}{\sqrt{2}}$, and $c_1, c_2,...
,c_{m}$ such that $p_{cov}(n,p,r_s,k)=c_i$ if $r_s \in [r_{i-1},r_i)$.
\end{thm}

\begin{proof}
For simplicity we prove the theorem for $k=1$; the extension to $k>1$ is straightforward. Let
$p_{cov}(r_s)=p_{cov}(n,p,r_s,1)$. We need to show $p_{cov}(r_s)$ ia
a piecewise constant function of $r_s$. It is clear that
$p_{cov}(r_s)$ is a nondecreasing function. In particular we have
$p_{cov}(0)=0$ and $p_{cov}(r_s)=1$ for $r_s \geq
\frac{1}{\sqrt{2}}$. For a point $X$ in the plane, let $circ(X,r)$ be
the closed ball that is centered at $X$ and has radius $r$. Define $cov(X,r_s)$
to  be the area that is covered by a sensor node located at $X$ with sensing radius $r_s$.
In other words, $cov(X,r_s)$ is
the portion of $circ(X,r_s)$ that lies within the deployment region.
Again assume that $\mathcal{X}=\{X_1,X_2,...,X_k\}$ is
the set of points in the sensor grid and $\mathcal{X}_a \subset \mathcal{X}$ is the set of active
sensors. Define
\begin{equation}\label{}
    cov(\mathcal{X}_a,r_s)=\bigcup_{X \in \mathcal{X}_a} cov(X,r_s).
\end{equation}
Thus the unit square $S_0$ is completely covered whenever
$cov(\mathcal{X}_a,r_s)=S_0$. Let $0 \leq r_s \leq \frac{1}{\sqrt{2}}$. If
$cov(\mathcal{X}_a,r_s)=S_0$, then for all $r>r_s$, we have
$cov(\mathcal{X}_a,r)=S_0$.  On the other hand, we prove that if
$cov(\mathcal{X}_a,r_s)\neq S_0$, there exists $\epsilon > 0$ such that for all $r \in
[r_s,r_s+\epsilon )$ we have $cov(\mathcal{X}_a,r)\neq S_0$. To prove this note that the
covered area $cov(\mathcal{X}_a,r_s)$ is a closed set because it is
the union of a finite number of closed sets. Thus, the uncovered
area  is an open set. Hence, to cover the uncovered area,
the sensing radius $r_s$ must increase by a strictly positive
amount.

We now prove that for any $r_s$, there exists a strictly positive
$\epsilon$ such that $p_{cov}(r)$ remains constant as the sensing radius $r$
varies within $[r_s,r_s+\epsilon)$. Define
\begin{equation}\label{eq:xa}
\begin{split}
    \mathbf{X}_{a}^{r_s}=\{ \mathcal{X}_a :
    cov(\mathcal{X}_a,r_s)=S_0\}, \quad \text{and}\quad
    \mathbf{\overline{X}}_{a}^{r_s}=\{ \mathcal{X}_a :
    cov(\mathcal{X}_a,r_s)\neq S_0\}.
\end{split}
\end{equation}
Note that $\mathbf{X}_{a}^{r_s}$ and $\mathbf{\overline{X}}_{a}^{r_s}$ are finite sets. Using (\ref{eq:xa}) we have
\begin{equation}\label{equ:pcov}
p_{cov}(r_s)= \sum_{\mathcal{X}_a \in \mathbf{X}_{a}^{r_s}}
p(\mathcal{X}_a).
\end{equation}
For any $\mathcal{X}_a \in \mathbf{\overline{X}}_{a}^{r_s}$ define $\epsilon(\mathcal{X}_a)=\textrm{min}\{r^\prime \ \textrm{s.th.} \ cov(\mathcal{X}_a,r_s+r^\prime)= S_0\}$ and let $\epsilon=\textrm{min} \{\epsilon(\mathcal{X}_a): \mathcal{X}_a \in
\mathbf{\overline{X}}_{a}^{r_s} \}$. Then $\epsilon>0$. Further,
for all $r \in [r_s,r_s+\epsilon)$, we have
$\mathbf{\overline{X}}_{a}^{r}=\mathbf{\overline{X}}_{a}^{r_s}$.
Thus we conclude that for all $r \in [r_s,r_s+\epsilon)$, we have
$\mathbf{X}_{a}^{r}=\mathbf{X}_{a}^{r_s}$. Using (\ref{equ:pcov}) we
conclude that $p_{cov}(r_s)$ does not change as $r$ varies in
$[r_s,r_s+\epsilon)$. This proves that $p_{cov}(r_s)$ is a right-continuous
piecewise constant function.

It remains to show that the number of discontinuities is finite. This follows easily from the fact that the number of
$\mathcal{X}_a$'s is finite. Note that by (\ref{equ:pcov}), any
discontinuity occurs when the set $\mathbf{X}_{a}^{r_s}$ changes
due to an increase in $r_s$. However, $\mathbf{X}_{a}^{r_s}$ can
have at most $2^n$ elements. Further, at each discontinuity, at least one
element is added to $\mathbf{X}_{a}^{r_s}$. This implies that the
number of discontinuities is upper-bounded by $2^n$. It is worth noting that
in practice, the number of discontinuities is much smaller than
$2^n$. This completes the proof.
\end{proof}

Theorems \ref{thm:piecegraph} and \ref{thm:piececoverage} determine
the behavior of a vast class of network quantities when they are
considered as functions of communication and sensing radii. In
particular, these are important from the view point of finite sensor
grids. We note that for very large network sizes, the piecewise
constant functions tend to continuous functions. Thus, we do not
observe the discontinuities. However, in such networks as finite sensor grids, this property is noticeable as in Figures
\ref{fig:coverageasymp} and \ref{fig:coverageasympn100p2k2}. We clarify that because the simulation
results are approximations for the actual values, the figures are not
completely piecewise constant. In fact one of the implications of
Theorems \ref{thm:piecegraph} and \ref{thm:piececoverage} is to
simplify simulations since the piecewise constant
functions can be completely determined by knowing their values for
only a finite number of points.
Furthermore, the above results suggest that increasing the
communication and sensing radii does not necessarily improve coverage,
connectivity or any other graph theoretic properties. This is an
important observation for designing the network and choosing its parameters optimally.


\subsection{Bounds on the Coverage Probability}
\label{sec:covanalysis}
We now consider coverage probability for finite sensor grids. We find lower and upper bounds
for $p_{cov}(n,p,r_s,k)$ and show that they can give an acceptable estimate of the coverage probability. Let $N(r,x,y)$ be the number of sensors whose distance from
the point $(x,y)$ is less than or equal to $r$. For example,
$N(r,0.5,0.5)$ denotes the number of sensors whose distance from the
top-right corner of unit square is less than or equal to $r$. We first prove the following lemma.

\begin{lem}\label{lem:lowercovgrid}
Let $L$ be the set of $\sqrt{l}\times \sqrt{l}$ points in a virtual grid on the unit square. Let us also denote by $A(u)$ the event that point $u$ on $L$ is covered by a sensor grid with coverage radius $r_s$. We then have
\begin{align}\label{eq:lowercovgrid}
\text{Pr}(\bigwedge_{u\in L}A(u))\geq \prod_{u\in L} \text{Pr}(A(u)).
\end{align}
\end{lem}

\begin{proof}
We use FKG inequality to prove this lemma \cite{FKG}.
We first show that for any two subsets $I$ and $J$ of $L$, we have $\text{Pr}(A(I) \bigwedge A(J))\geq \text{Pr}(A(I)) \times \text{Pr}(A(J))$. Here, $A(I)$ ($A(J)$) is the event that all points in $I$ ($J$) are covered.
Since this is true for any two subsets of $L$, (\ref{eq:lowercovgrid}) can be derived by partitioning $L$ and the resulted components, repeatedly, and then using this property at each step.

First note that we can enumerate the nodes in the sensor grid from 1 to $n$. Accordingly we can show the status of the network with a $n$-tuple binary vector where 0 and 1 are assigned to inactive and active nodes, respectively. Let us denote by $T$ the set of all possible binary $n$-tuples as the network status, i.e. $T=\{\textbf{t}=(t_1, t_2, ..., t_n) \in \{0,1\}^n\}$. $T$ can be then defined as a ``finite distributive lattice" as follows. For $\textbf{x}=(x_1, x_2, ..., x_n)$ and $\textbf{y}=(y_1, y_2, ..., y_n)$ in $T$, we define $\textbf{x}\vee \textbf{y}$ as the elementwise ``or" of $\textbf{x}$ and $\textbf{y}$. That is if  $\textbf{w}=\textbf{x}\vee \textbf{y}=(w_1, w_2, ..., w_n)$ then $w_i=x_i\vee y_i$.
Similarly, we define $\textbf{x}\wedge \textbf{y}$ as the elementwise ``and" of $\textbf{x}$ and $\textbf{y}$, i.e. if  $\textbf{w}=\textbf{x}\wedge \textbf{y}=(w_1, w_2, ..., w_n)$ then $w_i=x_i\wedge y_i$. With these definitions, it is easy to check that $\vee$ and $\wedge$ are distributive over each other. Note that the lattice defined this way is partially ordered as we have $(\textbf{x}\wedge \textbf{y}) \preceq \textbf{x},\textbf{y}\preceq (\textbf{x}\vee \textbf{y})$.

We now define a probability measure $\mu:T \rightarrow \mathbb{R}^+$ as follows. For $\textbf{x}\in T$, $\mu(\textbf{x})=p^k(1-p)^{n-k}$, where $k=\sum_{i=1}^n x_i$. Note that $\mu(\textbf{x})$ in fact indicates the probability that the sensor grid admits the status $\textbf{x}$ with $k$ active sensors and $n-k$ inactive sensors. It is also trivial to verify that $\mu(\textbf{x})\mu(\textbf{y})\leq \mu(\textbf{x}\vee \textbf{y})\mu(\textbf{x}\wedge \textbf{y})$, which is required by FKG inequality. Given two subsets $I$ and $J$ of $L$, we also define functions $f,g:T\rightarrow \mathbb{R}^+$ as follows. For every $\textbf{x}\in T$, $f(\textbf{x})=1$ ($g(\textbf{x})=1$) if $I$ ($J$) is covered. By these definitions, $f$ and $g$ are both \emph{increasing functions} over $T$. Given the lattice $T$, measure $\mu$, and functions $f$ and $g$ as above, the FKG inequality holds as follows.
\begin{align}\label{eq:FKG}
\big(\sum_{\textbf{x}\in T} \mu(\textbf{x})f(\textbf{x})\big) . \big(\sum_{\textbf{x}\in T} \mu(\textbf{x})g(\textbf{x})\big) \leq \big(\sum_{\textbf{x}\in T} \mu(\textbf{x})f(\textbf{x})g(\textbf{x})\big) . \big(\sum_{\textbf{x}\in T} \mu(\textbf{x})\big).
\end{align}
However, $\sum_{\textbf{x}\in T} \mu(\textbf{x})=1$. Furthermore, $\sum_{\textbf{x}\in T} \mu(\textbf{x})f(\textbf{x})$, $\sum_{\textbf{x}\in T} \mu(\textbf{x})g(\textbf{x})$, and $\sum_{\textbf{x}\in T} \mu(\textbf{x})f(\textbf{x})g(\textbf{x})$ are in fact equal to $\text{Pr}(A(I))$, $\text{Pr}(A(J))$, and $\text{Pr}(A(I) \bigwedge A(J))$, respectively. (\ref{eq:FKG}) can thus be rewritten as $\text{Pr}(A(I) \bigwedge A(J))\geq \text{Pr}(A(I)) \times \text{Pr}(A(J))$.
This completes the proof.
\end{proof}

Now, we are ready to prove the lower bound on the coverage probability.

\begin{thm}\label{thm:lowercovgrid}
Consider the coverage probability for a finite sensor grid with parameters $n, p$, and $r_s$. We then have
\begin{equation}\label{eq:lowercovgrid}
p_{cov}(n,p,r_s,1) \geq \prod_{u \in L} [1-(1-p)^{N(r_s^\prime, x_u, y_u)}],
\end{equation}
where $L$ is the virtual grid in Lemma \ref{lem:lowercovgrid}, and the radius $r_s ^\prime$ is given by $r_s-\frac{1}{\sqrt{2l}}$.
\end{thm}

\begin{proof}
First note that the choice of the virtual grid $L$ and its size, $l$, is arbitrary. As a result, for any given $r_s$, we choose $l$ large enough such that $r_s-1/\sqrt{2l} >0$.
To prove this theorem, we make use of some results in \cite{kcoverage2}. Lemma 3.1 in \cite{kcoverage2} states that for a given set of points $L$ that consists of all grid points of a $\sqrt{l}\times \sqrt{l}$ virtual grid on a unite square, if $L$ is covered by a network of radius $r_s ^\prime$, the unit square is covered by the same network but with the radius $r_s=r_s ^\prime+\frac{1}{\sqrt{2l}}$. Hence,  $p_{cov}(n,p,r_s,1) \geq \text{Pr}(L \ \text{covered})$. Now we use Lemma \ref{lem:lowercovgrid} above to prove the lower bound.
Let us denote by $A(u)$ the event that point $u$ is covered by a network with coverage radius $r_s ^\prime$. By Lemma \ref{lem:lowercovgrid} we have $\text{Pr}(\bigwedge_{u\in L}A(u))\geq \prod_{u\in L} \text{Pr}(A(u))$.
Now, note that the probability that a point $u$ with coordination $(x_u, y_u)$ is covered by the set of $n$ nodes with coverage radius $r_s ^\prime$ is given by $[1-(1-p)^{N(r_s^\prime, x_u, y_u)}]$. Thus, we write $p_{cov}(n,p,r_s,1) \geq \text{Pr}(L \  \text{covered})=\text{Pr}(\bigwedge_{u\in L}A(u)) \geq \prod_{u\in L} \text{Pr}(A(u))= \prod_{u\in L} [1-(1-p)^{N(r_s^\prime, x_u, y_u)}]$.
\end{proof}

Now, we prove an upper bound for the coverage probability.

\begin{thm} \label{thm:uppercov}
Consider sensor grids with parameters $n$, $p$, $r_s$. Then the coverage probability is upper bounded by
\begin{align}\label{eq:uppcov}
\nonumber p_{cov}(n,p,r_s,1) \leq [1-(1-p)^{N(r_s,.5,.5)}]^{4}
&\times [1-(1-p)^{N(r_s,.5,0)}]^{4 \lfloor \frac{(1-2r_s)}{2 r_s}\rfloor}\\
&\times [1-(1-p)^{N(r_s,0,0)}]^{{\lfloor \frac{(1-2r_s)}{2
r_s}\rfloor}^2},
\end{align}
where $\lfloor x \rfloor$ denotes the largest integer less than or equal to $x$.
\end{thm}
\begin{proof}
Let $X_1=(x_1,y_1), X_2=(x_2,y_2),..., X_m=(x_m,y_m)$ be $m$ points
on the deployment region $S_0$. Assume that $d(X_i,X_j)>2r_s$ for
$i\neq j$, where $d(.,.)$ is the Euclidean distance between the
points. Then the event that $X_i$ is covered is independent of the
event that $X_j$ is covered. This is because there is no sensor node
that can cover both points. Hence, the probability
that all $X_i$'s are covered is given by
\begin{equation}\label{eq:upcoveval}
    \prod_{i=1}^{m}[1-(1-p)^{N(r_s,x_i,y_i)}].
\end{equation}
This implies that $p_{cov}(n,p,r_s,1)$ is upper bounded by
$\prod_{i=1}^{m}[1-(1-p)^{N(r_s,x_i,y_i)}]$. Thus using any set of
points on the plane that satisfy $d(X_i,X_j)>2r_s$, we can find an
upper bound for $p_{cov}(n,p,r_s,1)$. In particular, considering the
set of points given by Figure \ref{fig:upperboundcov}, we obtain the
upper bound in (\ref{eq:uppcov}).
\end{proof}

\begin{figure}[h]
\centering
\includegraphics[width=3 in, height=2.7 in]{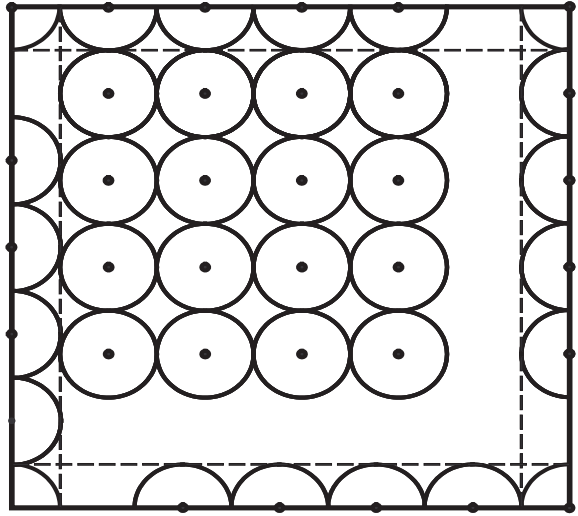}
\caption{Location of the points used for the upper bound of theorem \ref{thm:uppercov}. The centers of the circles are points $X_1=(x_1,y_1), X_2=(x_2,y_2),..., X_k=(x_k,y_k)$.}
 \label{fig:upperboundcov}
\end{figure}

Note that the choice of $X_i$'s in the proof ensures that we consider the edge effects. In
fact, in many situations, the coverage probability is dominated by
the first and second terms in (\ref{eq:uppcov})
which are related to edge effects. One may suggest that using a triangular grid, instead of non-overlapping balls, can result in a more dense packing and consequently a better bound. However, using a triangular grid results in fewer nodes on the sides of the square. We have evaluated (\ref{eq:upcoveval}) for the triangular grid as well as some other more complicated layouts. It turned out that the resulting bound is looser for the triangular grid. Moreover, there is only a negligible improvement by using other layouts at the expense of a more complicated expression compared to (\ref{eq:uppcov}). It is also worth noting that
$N(r_s,.5,.5)$, $N(r_s,.5,0)$, and $N(r_s,0,0)$ introduce
discontinuities in the upper bound as predicted by Theorem
\ref{thm:piececoverage}.

Figure \ref{fig:coveragefinite} compares the results obtained by Theorems \ref{thm:lowercovgrid} and \ref{thm:uppercov} and the simulations for $n=100$ and $p=0.2$.
We observe that Theorems \ref{thm:lowercovgrid} and \ref{thm:uppercov} provide
significantly better estimates of coverage probability compared with
the asymptotic analysis in Figure \ref{fig:coverageasymp}.
The asymptotic behavior of these bounds can be checked by letting $n$ to grow large. The derivation of the lower bound employs a similar argument as in the case of Lemma 4.1 in \cite{kcoverage2}. It can be checked that this bound leads to the same asymptotic expression as in Theorem \ref{KLB}, hence it is tight asymptotically.
On the other hand, when $n$ gets large, we can reasonably expect the same situation as in the upper bound of Theorem \ref{thm:uppercov}. That is the terms corresponding to the virtual nodes (nodes on the virtual grid) on the corner and close to the edges will be dominant in the lower bound of (\ref{eq:lowercovgrid}). This is true because there are fewer sensor nodes around these virtual nodes to cover them, causing the coverage probability $ \text{Pr}(A(u))$ for these virtual nodes to decay faster than the rest of the virtual nodes.
Regarding the asymptotic behavior of the upper bound of (\ref{eq:uppcov}), it can be verified that if $\frac{np\pi r^2}{\log(np)}<1-\epsilon$ as $n$ tends to infinity, then the upper bound will be $o(1)$.

We also like to talk about the time complexity of computing the bounds in Theorems \ref{thm:lowercovgrid} and \ref{thm:uppercov}. The upper bound of Theorem \ref{thm:uppercov} can be computed in time $O(n)$. This is because we need to find the neighbors for $1/(2r_s )^2$ points, and finding the number of neighbors for each point takes a constant amount of time. However, if $1/(2r_s )^2 >n$, then $p_{cov}=0$. Thus, the complexity is $O(1/(2r_s )^2 )=O(n)$. For the lower bound in Theorem \ref{thm:lowercovgrid}, note that we only need to find the number of neighbors for every node of $L$. Given the sensor grid and the virtual grid $L$, finding the number of neighbors for each node of $L$ takes a constant amount of time. Also note that for $r_s^{\prime}$ to be positive, $L$ needs to contain more nodes than the sensor grid, hence, $l>n$. Therefore, the lower bound can be computed with complexity $O(l)$.

Theorems \ref{thm:lowercovgrid} and \ref{thm:uppercov} can be easily
generalized for $k$-coverage. Since the proof is very similar, we just
state the result in one theorem.

\begin{figure}[t]
\centering
\includegraphics[width=4 in, height=3 in]{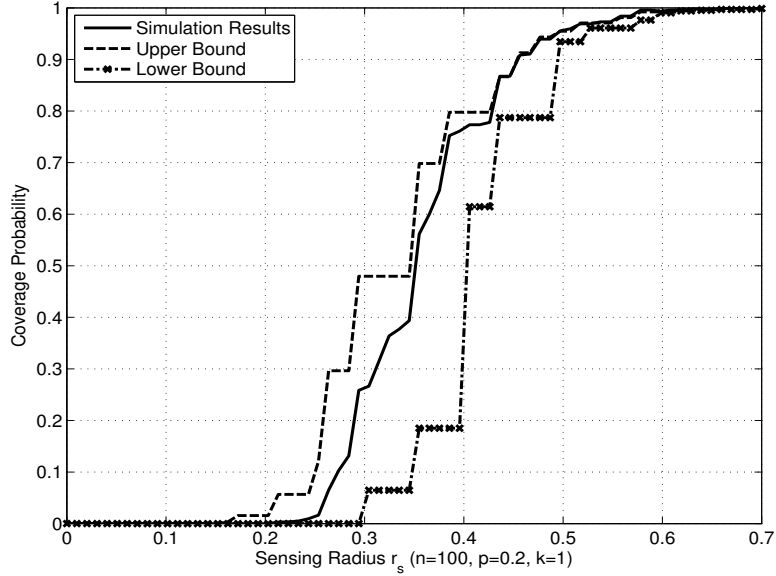}
\caption{Comparison of finite-size analysis with the simulation
results for the coverage probability of sensor grids with $n=100$
and $p=0.20$.}
\label{fig:coveragefinite}
\end{figure}

\begin{thm} \label{thm:upperkcov}
Consider the $k$-coverage probability for a sensor grid with parameters $n$, $p$, and $r_s$, and assume that $L$ and $r_s ^\prime$ are as defined in Theorem \ref{thm:lowercovgrid}. Then we have
\begin{equation}
p_{cov}(n,p,r_s,k) \geq \prod_{u \in L} [1-\sum_{i=0}^{k-1}{N(r_s ^\prime,x_u,y_u)\choose i}p^i(1-p)^{N(r_s ^\prime,x_u,y_u)-i}],
\end{equation}
and
\begin{align}\label{eq:uppkcov}
\nonumber p_{cov}(n,p,r_s,k) \leq & [1-\sum_{i=0}^{k-1}{N(r_s,.5,.5)\choose i}p^i(1-p)^{N(r_s,.5,.5)-i}]^{4}
\times \\ \nonumber & [1-\sum_{i=0}^{k-1}{N(r_s,.5,0)\choose i}p^i(1-p)^{N(r_s,.5,0)-i}]^{4 \lfloor \frac{(1-2r_s)}{2 r_s}\rfloor}
\times \\ & [1-\sum_{i=0}^{k-1}{N(r_s,0,0)\choose i}p^i(1-p)^{N(r_s,0,0)-i}]^{{\lfloor
\frac{(1-2r_s)}{2 r_s}\rfloor}^2}.
\end{align}
\end{thm}

\vspace{.2 in}
Figure \ref{fig:coveragefiniten400p2k2} compares the results obtained by
Theorem \ref{thm:upperkcov} and the simulations for $k=2$, $n=100$, and
$p=0.2$. We observe that the two results are very close.

\begin{figure}[t]
\centering
\includegraphics[width=4 in, height=3 in]{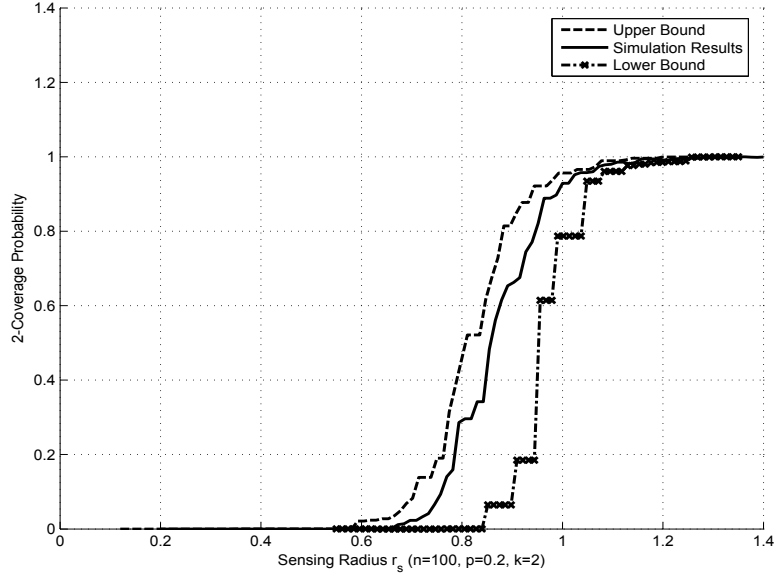}
\caption{Comparison of finite-size analysis with the simulation
results for the 2-coverage probability ($k=2$) of sensor grids with
$n=100$ and $p=0.20$.}
\label{fig:coveragefiniten400p2k2}
\end{figure}

\section{Small-Scale Analysis for Random Sensor Networks}\label{sec:general}
In this section, we try to establish a framework for analysis of
finite sensor networks with random node deployment.
As we mentioned earlier, the exact
analysis of network properties is usually very difficult or at least results in very
complicated formulas. Thus, we will try to find simple lower and
upper bounds which are sufficiently close together that can be used to find a good estimate of the exact value of the desired property. Here, we consider coverage and connectivity in finite sensor networks.

\subsection{Preliminaries}\label{sec:formPrelim}
We consider a wireless sensor network that consists of $n$ nodes and
assume that the nodes are placed on a plane based on a given
probability distribution. For example, in wireless sensor networks
it is usually assumed that the nodes are randomly and uniformly
deployed over a given field \cite{akyildiz02survey}. We assume that
each node has a fixed communication radius. Two nodes are
connected (can communicate with each other) if they are within
communication range of each other. Throughout the paper, we assume
$\mathcal{B(\mathbb{R}^{\mathrm{2}})}$ is the Borel $\sigma-$algebra
on $\mathbb{R}^2$ and $m$ is the Lebesgue measure on
$\mathcal{B(\mathbb{R}^{\mathrm{2}})}$. Note that we just use
measure theoretic definitions to take care of technicalities but it
is not necessary for the reader to be familiar with them. The reader
can simply assume that for a set $F$ in $\mathbb{R}^{2}$, $m(F)$ is
the area of $F$. $B(X,R)$ is the closed ball
with radius $R$ centered at $X$ in $\mathbb{R}^{2}$.
$S(X,L)$ is the closed square with side $L$
centered at $X$ in $\mathbb{R}^{2}$. In particular
$S_0=S(O,1)$ is the closed square with unit
area centered at the origin. If $u$ and $v$ are two nodes of a
network located in $\mathbb{R}^{\mathrm{2}}$, then $d(u,v)$ is the
Euclidean distance between the location of the points. For any set
$F\in \mathcal{B(\mathbb{R}^{\mathrm{2}})}$ we define
$\nu(F)=m(F\cap S_0)$. Clearly, $\nu$ defines a measure on
$\mathcal{B(\mathbb{R}^{\mathrm{2}})}$. 

Wireless networks are sometimes modeled with the probability space
of graphs that we represent with $g(n,r)=g(n,r(n))$. In this model,
it is assumed that $n$ nodes are uniformly and randomly distributed
over $S_0=S(O,1)$. If two nodes $u$ and $v$
satisfy $d(u,v)\leq r(n)$, then the edge $\{u,v\}$ belongs to edges
of the graph. A more general model is the model $g(n,r,p)$, in which
two nodes are connected with probability $0<p\leqslant 1$ if their
distance is less than $r$.
In this model $p$ models link failures
that are common in wireless networks.
Note that here we are using $p$ as a different notation from the previous section.
 Asymptotic properties of
$g(n,r)$ have been studied extensively. Here we are interested in these properties when $n$ is not necessarily
large. It is worth noting that the assumption that the nodes are distributed on a square is made for simplicity.
These arguments can easily be generalized to other models for the deployment region as well as the case where
nodes are distributed non-uniformly over the deployment region. For the purpose of analysis, we divide the
square $S_0$ to different parts shown in Figure \ref{square}.

\begin{figure}[t]
\centering
\includegraphics[width=.5 \columnwidth,angle=0]{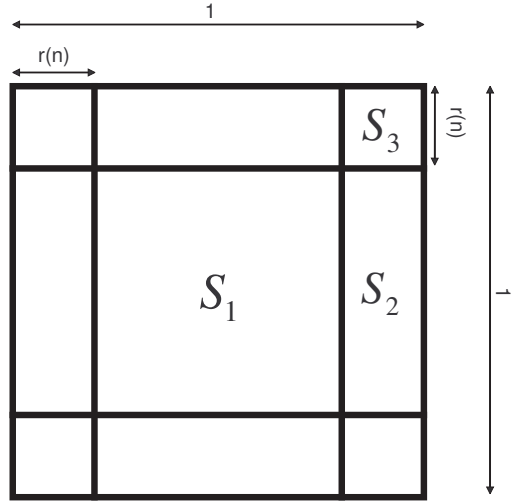}
\caption{The field $S_0$ and its subdivisions $S_1$,$S_2$, and
$S_3$.}
 \label{square}
\end{figure}

\subsection{Asymptotic versus Finite Analysis} \label{sec:need}
In this section, we present some evidence to show that previous
asymptotic results diverge significantly from actual values for
finite networks. To show this, we consider connectivity.
We first
provide the asymptotic probability of disconnectivity for $g(n,r,p)$
and compare it to simulation results. The
following result is proved in \cite{gupta98critical}, where a
slightly different model is considered. However, the results can be
trivially extended to $g(n,r)$.

\begin{thm} (Gupta and Kumar 1998)
Let $c_n=n \pi r^2-\log(n)$, then $g(n,r)$ is connected with high
probability if $\lim \limits_{n \rightarrow \infty} c_n= \infty$. On
the other hand, if $\lim \limits_{n\rightarrow \infty} c_n=c<\infty$
then for large $n$, $g(n,r)$ is disconnected with a strictly
positive probability $1-p_{asymp}(c)$.
\end{thm}

This theorem states that if $\lim \limits_{n\rightarrow \infty}
c_n=c<\infty$, the network connectivity probability will be bounded
away from one. In fact, $p_{asymp}(c)$ is the limit for the
probability that the network is connected when $n$ goes to infinity.
To find $p_{asymp}(c)$, Penrose in \cite{penrose97} proved that $g(n,r)$ is connected if and only if  the longest edge of its corresponding Minimal Spanning Tree (MST) is smaller than $r$. On the other hand, if we denote the longest edge of the MST by $M_n$, it is shown in \cite{penrose97} that the distribution of $n \pi {M_n}^2 -\log n$ converges to the double exponential distribution:
\begin{equation}
\lim_{n \to \infty} P[n \pi {M_n}^2 -\log n \leq \alpha]=exp(- e^{-\alpha}) \quad \textrm{for} \ \alpha \in \mathbb{R}.
\end{equation}
Thus we have
\begin{equation}\label{eqn:pasymp}
p_{asymp}(c)=\lim_{n \to \infty} P[M_n \leq r]=\lim_{n \to \infty} P[n \pi {M_n}^2 -\log n \leq n \pi r^2 -\log n]=e^{-e^{-c}}.
\end{equation}
Therefore, asymptotically, the probability that $g(n,r)$ is
connected is given by
\begin{align}
\nonumber p_{asymp}=e^{-ne^{-n \pi r^2}}.
\end{align}
In Figure \ref{fig:asmCompare}, we compare the probability of having
a disconnected graph for $n = 100$ and for both exhaustive
simulations and the asymptotic results. In Figure \ref{fig:asmCompare}, the
probability of disconnectivity is shown as a function of $r$, the
communication radius. The experiment shows that these results may
differ by 10 orders of magnitude. This illustrates that the
asymptotic method fails to provide a good approximation for small-
scale networks.

\begin{figure}[t]
\centering
\includegraphics[width = .65 \columnwidth]{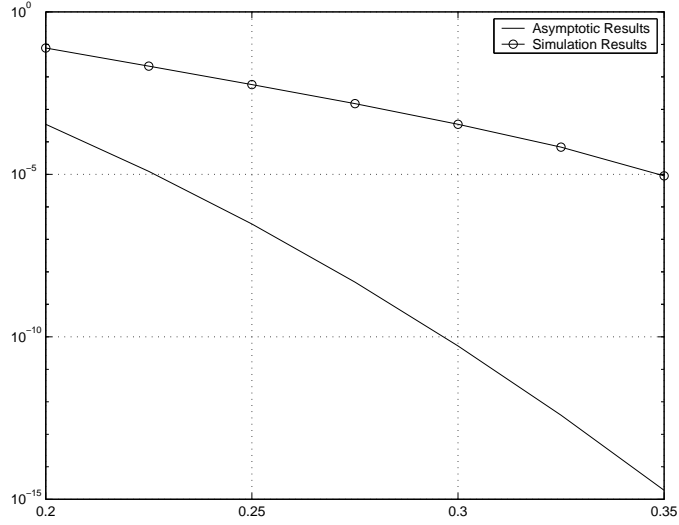}
 \caption{Comparison of asymptotic results with the small scale
simulation results for the probability of disconnectivity of
$g(n=100,r)$.}
 \label{fig:asmCompare}
\end{figure}

A natural question to ask is what makes the results for the asymptotic analysis so different from the finite case? As you can see in Figure ~\ref{square}, $S_0$ is formed by three regions, $S_1$, and boundary regions $S_2$ and $S_3$.
One important phenomenon in asymptotic analysis is that boundary
effects can be neglected. Loosely speaking, the asymptotic analysis of the
network properties is usually dominated by what happens in region
$S_1$ in Figure~\ref{square}. This can considerably simplify the analysis and
results in simple and closed-form formulas for network properties. In fact, we saw an example of this
phenomenon in the asymptotic formula for connectivity in (\ref{eqn:pasymp}).
However, in small-scale networks boundary effects cannot be
neglected. In other words, nodes in the corners of the field can
play an important role in some network properties.

Another important issue in the analysis of finite networks is the effect of constant factors.
In asymptotic analysis, we usually neglect constant factors.
However, in the small-scale analysis, we must consider them. This is in
fact a distinction of any finite analysis from the asymptotic
analysis and is not specific to geometric graphs.

\subsection{Small-Scale Analysis for Coverage}
In this section, we study the coverage probability, $p_{cov}(n,r)$, for finite sensor networks modeled by $g(n,r)$. We prove lower and upper bounds for the coverage probability. We start with the lower bound which gives the worst case performance as well as a guarantee of the coverage probability.

\begin{thm}\label{th:lowcov}
Consider the coverage probability of a sensor network modeled by $g(n,r)$. Then we have
\begin{equation}
p_{cov}(n,r) \geq 1-\sum_{u\in L} [1-\nu(B(X(u),r^\prime))]^n,
\end{equation}
where $L$ is the set of $\sqrt{l}\times \sqrt{l}$ points in a virtual grid on the unit square and the radius $r^\prime$ is given by $r-\frac{1}{\sqrt{2l}}$.
\end{thm}

\begin{proof}
We briefly describe the proof. As in the case of grid deployment, showing that $L$ is covered guarantees the coverage of the entire region.
Let $A(u)$ be the event that the virtual grid point $u$ is covered. Using union bound, we have
\begin{align}
\text{Pr}(L \ \text{not covered})=\text{Pr}\big[\bigcup_{u\in L}\overline{A(u)}\big] \leq \sum_{u\in L} \text{Pr}(\overline{A(u)})=\sum_{u\in L} [1-\nu(B(X(u),r^\prime))]^n.
\end{align}
Therefore, $p_{cov}(n,r)=1-\text{Pr}(L \ \text{not covered})\geq 1- \sum_{u\in L} [1-\nu(B(X(u),r^\prime))]^n$.
\end{proof}

Now we prove an upper bound for the coverage probability.

\begin{thm}\label{th:uppcovrand}
The coverage  probability of a unit square for a sensor network modeled by $g(n,r)$ has an upper bound given by
\begin{align}\label{eq:uppcovrand}
p_{cov}(n,r)  \leq  [1-(1-\frac{\pi r^2}{4})^n]^{4}
\times  [1-(1-\frac{\pi r^2}{2})^n]^{4 \lfloor \frac{(1-2r)}{2 r}\rfloor}
 \times [1-(1-\pi r^2)^n]^{{\lfloor \frac{(1-2r)}{2 r}\rfloor}^2}.
\end{align}
\end{thm}

\begin{proof}
We adapt the proof of Theorem \ref{thm:uppercov} to prove Theorem \ref{th:uppcovrand}. Consider $k$ points $U_1, U_2, ..., U_k$ on the unit square and assume that these $k$ points are at least apart by $2r$ units from one another. Similar to the proof of Theorem \ref{thm:uppercov}, we can observe that $p_{cov}(n,r)$ is upper bounded by the probability that all the $k$ points are covered which is given by $\prod_{i=1}^{k} (1-[1-\nu(B(U_i,r))]^n)$. Using the set of points depicted in Figure \ref{fig:upperboundcov}, we find the upper bound given by (\ref{eq:uppcovrand}).
\end{proof}

Figure \ref{fig:covrand} compares the bounds predicted by Theorems \ref{th:lowcov} and \ref{th:uppcovrand} with the simulated coverage probability value. Asymptotic result from \cite{kcoverage2} is also presented. Clearly, the bounds are more useful than the asymptotic result in the sense that they give a better estimate of the coverage probability.

\begin{figure}[t]
\centering
\includegraphics[width=4 in, height=3 in]{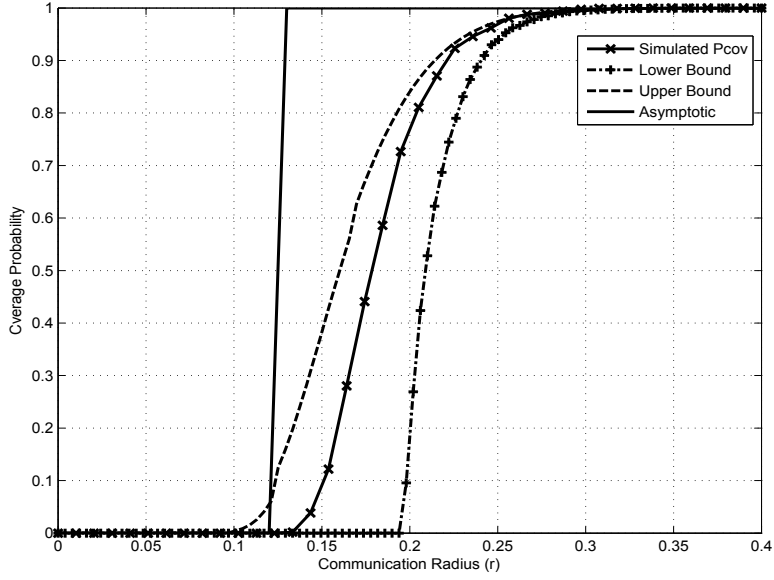}
\caption{Simulation results, and upper and lower bounds of coverage probability for a random wireless network of size $n=100$.}
\label{fig:covrand}
\end{figure}

\subsection{Small-Scale Analysis for Connectivity}
In this section, we study the connectivity properties of finite sensor networks modeled by $g(n,r,p)$. We find lower and upper
bounds for the probability $p_{disc}(n,r,p)$ that $g(n,r,p)$ is disconnected. Let $p_{low}(n,r,p)$ and $p_{upp}(n,r,p)$ be the
lower and upper bounds on $p_{disc}(n,r,p)$, respectively.
Here we consider the case where $p_{disc}(n,r,p)$ is small, i.e.,
$p_{disc}(n,r,p)<0.1$. In practice, this is usually the range that is important, since we want the network to be connected with high
enough probability. Using these bounds, we then provide a simple formula to estimate $p_{disc}(n,r,p)$.
As we will see by simulations, the proposed formula gives a very good estimate for $p_{disc}(n,r,p)$.
First, note that a connected component of a graph $g$ is defined as a connected subgraph that is isolated from the rest of $g$.

\begin{thm}
Consider a wireless sensor network modeled by $g(n,r,p)$. Then we have
\begin{align}\label{eq:1low}
    \nonumber p_{disc}(n,r,p) \geq & n  \int \limits_{S_0} \bigg(1-\nu(B(X,r))p \bigg)^{n-1}d m(X)\\
    \nonumber &-\binom {n}{2} \int \limits_{S_0}  \int \limits_{S_0} \bigg(1-\nu(B(X,r))p-\nu(B(Y,r))p \\
    & +\nu(B(X,r) \cap B(Y,r))p^2 \bigg)^{n-2}d m(X)\times m(Y),
\end{align}
and
\begin{align}\label{eq:1up}
\nonumber p_{disc}(n,r,p)\leq & \sum_{k=1}^{n/2}\binom {n}{k} p_{comp}(\{v_1, v_2, ..., v_k\})= n \int \limits_{S_0} \bigg(1-\nu(B(X,r))p \bigg)^{n-1}d m(X)\\
&+ \sum_{k=2}^{n/2}\binom {n}{k} p_{comp}(\{v_1, v_2, ..., v_k\}),
\end{align}
where $p_{comp}(\{v_1, v_2, ..., v_k\})$ is the probability that the vertices in $\{v_1, v_2, ..., v_k\}$ construct a connected component in $g(n,r,p)$.
\end{thm}

\begin{proof}
Let $p_{1}(n,r,p)$ be the probability that there exists at least one
isolated node (a vertex with no neighbors) in $g(n,r,p)$. Let also
$v_1, v_2, ..., v_n$ be the $n$ vertices of $g(n,r,p)$. Then
$p_{disc}(n,r,p)\geq p_{1}(n,r,p)$. Applying the inclusion-exclusion
lemma we obtain
\begin{align}\label{eq:lowcon}
     \nonumber p_{disc}(n,r,p) \geq &
     \nonumber \sum_{k=1}^{n} (-1)^{k+1} \binom {n}{k}\ \textrm{Pr}
     \{v_1, v_2, ..., v_k  \ \textrm{are isolated vertices} \} \\
      \geq &\nonumber  n \ \textrm{Pr}   \{v_1 \ \textrm{is isolated} \}  - \binom
     {n}{2} \textrm{Pr} \nonumber \{v_1 \ \textrm{and} \ v_2 \ \textrm{are isolated} \}.
\end{align}
Note that $\textrm{Pr} \{v_1 \ \textrm{is isolated } \}= \int \limits_{S_0} \bigg(1-\nu(B(X,r))p \bigg)^{n-1}dm(X)$.
Now define $Circ(a,b,r)=\{(x,y): (x-a)^2+(y-b)^2 \leq r^2 \}$.
Then we have
\begin{eqnarray}\label{eq:2isolated}
\nonumber && \textrm{Pr} \{v_1 \ \textrm{and} \ v_2 \ \textrm{are isolated vertices} \} =
\int \limits_{S_0}  \int \limits_{S_0 \setminus Circ(X,r)} \bigg(1-\nu(B(X,r))p-\nu(B(Y,r))p+ \\
\nonumber &&\nu(B(X,r) \cap B(Y,r))p^2 \bigg)^{n-2}d m(X)\times m(Y)+(1-p)\times\\
\nonumber && \int \limits_{S_0}  \int \limits_{Circ(X,r)} \bigg(1-\nu(B(X,r))p-\nu(B(Y,r))p+ \nonumber\nu(B(X,r) \cap B(Y,r))p^2 \bigg)^{n-2}d m(X)\times m(Y)\\
&&\leq \int \limits_{S_0}  \int \limits_{S_0} \bigg(1-\nu(B(X,r))p-\nu(B(Y,r))p+ \nonumber\nu(B(X,r) \cap B(Y,r))p^2
\bigg)^{n-2}d m(X)\times m(Y).
\end{eqnarray}
Combining these equations, we conclude the lower bound.
For the upper bound, note that $p_{disc}(n,r,p)$ is equal to the probability that $g(n,r,p)$
has at least one component of size less than $n/2$. This is given by eq. (\ref{eq:1up}).
\end{proof}

Note that the bounds for $p_{disc}(n,r,p)$ may not satisfy the simplicity requirement. Particularly in the upper bound, except the first few terms,
finding the rest of them is computationally infeasible. We now try to give an estimation of $p_{disc}(n,r,p)$ based on these bounds.
Let us denote the $k$th term in the upper bound by $a_k$.
We recall the assumption that $p_{disc}(n,r,p)$ is not very large, specifically we assumed $p_{disc}(n,r,p)<0.1$. An important observation here is that, by this
assumption, the $a_k$ coefficients decay very fast and, hence, the term $\sum_{k=1}^{n/2} a_k$ is dominated by $a_1$. This
can be seen by both numerical simulations and intuitive analytical arguments. In fact, as it is shown in \cite{Randomgeographspenrose}, as $n$ tends to infinity, the impact of the terms $a_k$, $k>1$ fades.
Figure \ref{neglig} compares $a_1$ and $a_2$ for $g(n=100,r,p=0.5)$. As we see $a_2$ is at least one order of magnitude smaller than $a_1$.
Using the same approach, we find out that a similar argument is true about the first and second terms in the lower bound of eq. (\ref{eq:1low}).
However, the first term is shared by both the lower and upper bounds.
Based on these observations we approximate the probability of disconnectivity as follows.
\begin{equation}\label{eq:discon}
p_{disc}(n,r,p) \simeq n  \int \limits_{S_0}
\bigg(1-\nu(B(X,r))p \bigg)^{n-1}d m(X).
\end{equation}
Figure \ref{conlowup} shows the upper bound, lower bound, and the simulation result for the probability of disconnectivity of
$g(n,r,p)$, for $n=100$, and $p=0.5$. As it can be seen, the three curves almost overlap. Based on our simulations, similar results are achieved if we
use different choices of parameters.

\begin{figure}[t]
\centering
\includegraphics[width=.65 \columnwidth,angle=0]{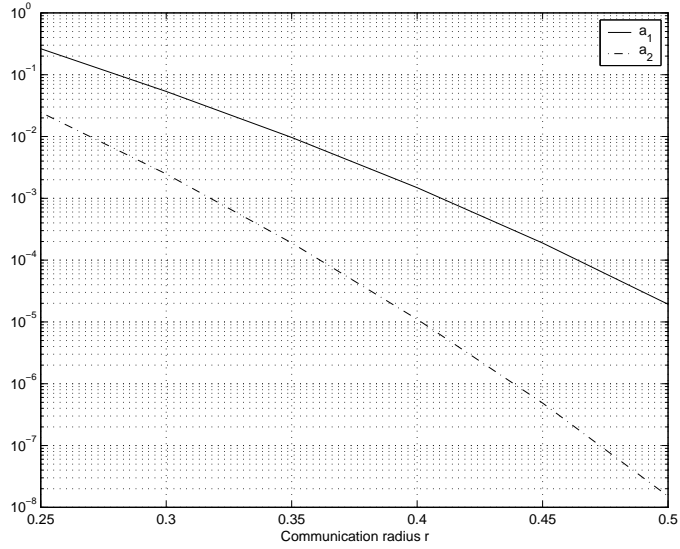}
\caption{Comparison of $a_1$ and $a_2$ in (\ref{eq:1up}).}
 \label{neglig}
\end{figure}

\begin{figure}[t]
\centering
\includegraphics[width=.65 \columnwidth,angle=0]{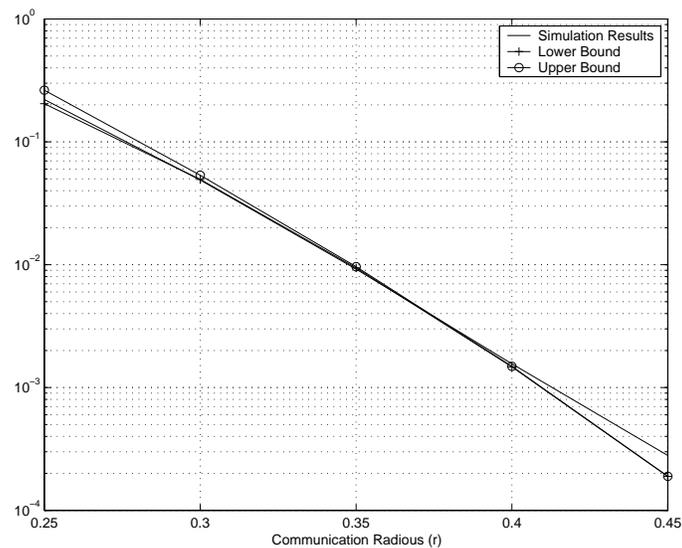}
\caption{Disconnectivity probability of $g(100,r,.5)$: lower bound,
upper bound, and the simulation results.}
 \label{conlowup}
\end{figure}

It is worth noting that the methodology used here can be used to
study $k$-connectivity. In summary, we find the following approximation of the probability that $g(n,r)$ is not $k$-connected
\begin{eqnarray}\label{eqn:summary}
p_{k,disc}(n,r) \simeq \sum_{j=0}^{k-1} n \binom{n}{j}
\int \limits_{S_0}
 [\nu(B(X,r(n)))]^j \times \bigg(1-\nu(B(X,r(n)))\bigg)^{n-j-1}d m(X).
\end{eqnarray}
Our simulations for different values of $k$ confirm the
validity of (\ref{eqn:summary}). Here, due to the space limitations
we omit those results.

\section{Conclusions} \label{sec:conclusion}
In this paper, we took some initial steps towards analyzing finite wireless sensor
networks. We provided some compelling evidence to show that
asymptotic results are not suitable for analyzing practical finite sensor
networks.
We studied connectivity and coverage of finite unreliable sensor grids as a special case. We showed that the connectivity, as well as all the graph theoretic quantities, are piecewise constant functions of the transmission radius in such networks. We also proved that the coverage has a similar behavior. Moreover, we obtained lower and upper bounds for the coverage and $k$-coverage probability of the grids and
verified their preciseness through simulations. Next, we extended our study to finite sensor networks with random node deployment.
Specifically, we considered coverage and connectivity of such networks. We derived lower and upper bounds for their coverage and showed how they can be used to estimate the coverage probability of the network. We also obtained a formula for
connectivity of wireless sensor networks and verified its accuracy through simulations. The formula was then extended to include $k$-connectivity.
A common characteristic of all these bounds is the ease of computations, making them very attractive.

This paper also opens up many research possibilities that offer some potentials
for further study.  In the past, many other important properties
of wireless sensor networks have been studied for large-scale networks. It
is an important task to extend these results for networks with
practical sizes, i.e. small-scale networks. Small-scale analysis can also reveal the effects of network parameters on network characteristics. The next step would be to derive more accurate bounds for network parameters such as coverage, connectivity, and MAC layer capacity and further use the
small-scale framework in the design, analysis, and evaluation of
communication algorithms for wireless networks.

\bibliographystyle{acmsmall}
\bibliography{bib1,bib2,bib3,bib4}

\end{document}